\documentclass[conference]{IEEEtran}

\IEEEoverridecommandlockouts

\usepackage{amsmath,amssymb,amsfonts}
\usepackage{mathtools}
\usepackage{mathabx}
\usepackage{amsthm}
\usepackage{nicefrac
}
\newtheorem{theorem}{Theorem}

\usepackage{algorithmic}
\usepackage{graphicx}
\usepackage{textcomp}
\usepackage{xcolor}
\def\BibTeX{{\rm B\kern-.05em{\sc i\kern-.025em b}\kern-.08em
    T\kern-.1667em\lower.7ex\hbox{E}\kern-.125emX}}

\usepackage{subfig}
\usepackage{tikz}
\usepackage[nolist,nohyperlinks]{acronym}
\usepackage{siunitx}
\usepackage{cite}
\usepackage{microtype}
\usepackage{nicefrac}
\renewcommand{\baselinestretch}{0.945} 
\usepackage{microtype}
\newcommand{\circnum}[1]{%
  \tikz[baseline=(n.base)]\node[circle,draw,inner sep=1pt,line width=0.5pt,font=\small](n){#1};%
}

\makeatletter
\newcommand\footnoteref[1]{\protected@xdef\@thefnmark{\ref{#1}}\@footnotemark}
\makeatother

\DeclareMathOperator*{\argmax}{arg\,max}	
\newcommand{\defeq}{\vcentcolon=}

\acrodef{0-D}[0-D]{zero-dimensional}
\acrodef{MC}[MC]{molecular communication}
\acrodef{DF}[DF]{decision-feedback}
\acrodef{CVS}[CVS]{cardiovascular system}
\acrodef{CIR}[CIR]{channel impulse response}
\acrodef{FPT}[FPT]{first passage time}
\acrodef{IG}[IG]{inverse Gaussian}
\acrodef{ISI}[ISI]{inter-symbol interference}
\acrodef{ML}[ML]{maximum-likelihood}
\acrodef{MWC}[MWC]{multipath wireless communications}
\acrodef{MIGHT}[MIGHT]{mixture of inverse Gaussians for hemodynamic transport}
\acrodef{OOK}[OOK]{on-off keying}
\acrodef{PDF}[PDF]{probability density function}
\acrodef{PDP}[PDP]{power delay profile}
\acrodef{RV}[RV]{random variable}
\acrodef{RMS}[RMS]{root mean squared}
\acrodef{SNR}[SNR]{signal-to-noise ratio}
\acrodef{Tx}[Tx]{transmitter}
\acrodef{Rx}[Rx]{receiver}
\acrodef{UCA}[UCA]{uniform concentration assumption}
\acrodef{VN}[VN]{vessel network}
\acrodef{SER}[SER]{symbol error rate}
\acrodef{wrt}[w.r.t.]{with respect to}
\acrodefplural{IG}[IGs]{inverse Gaussians}
\acrodefplural{PDF}[PDFs]{probability density functions}
\acrodefplural{RV}[RVs]{random variables}
\acrodefplural{SNR}[SNRs]{signal-to-noise ratios}
\acrodefplural{CIR}[CIRs]{channel impulse responses}
\acrodefplural{VN}[VNs]{vessel networks}
\acrodefplural{SER}[SERs]{symbol error rates}

\begin{document}

\title{%
Multipath Channel Metrics and Detection in \\Vascular Molecular Communication: \\A Wireless-Inspired Perspective
\thanks{\textsuperscript{*}Co-last authorship.\\
This work was funded in part by the German Federal Ministry of Research, Technology and Space (BMFTR) through Project Internet of Bio-Nano-Things (IoBNT) -- grant number 16KIS1987, in part by the German Research Foundation (Deutsche Forschungsgemeinschaft, DFG) under GRK 2950 -- ProjectID 509922606 and under grant number SCHA 2350/2-1, in part by the European Union’s Horizon Europe -- HORIZON-EIC-2024-PATHFINDEROPEN-01 under grant agreement Project N. 101185661, and in part by the Horizon Europe Marie Skodowska Curie Actions (MSCA)-UNITE under Project 101129618.}
}

\author{\IEEEauthorblockN{\scalebox{.99}{%
  Timo Jakumeit\textsuperscript{1}, Lukas Brand\textsuperscript{1}, Josep M.~Jornet\textsuperscript{2}, Robert Schober\textsuperscript{1}, Maximilian Sch\"afer\textsuperscript{1,*}, and Sebastian Lotter\textsuperscript{1,*}%
}}\\[-0.4cm]
\IEEEauthorblockA{\small \textsuperscript{1}Friedrich-Alexander-Universität Erlangen-Nürnberg (FAU), Erlangen, Germany\\
\textsuperscript{2}Northeastern University, Boston, MA, USA}
}

\maketitle

\begin{abstract}
Motivated by classical communications engineering, early works in \ac{MC} largely adopted established modeling and signal processing concepts from wireless electromagnetic communication systems. In the context of the human \ac{CVS}, \ac{MC} channel models evolved from simple unbounded and single-duct environments mimicking individual blood vessels to complex \ac{VN} topologies, generally at the expense of analytical tractability. Up until now, this has largely prohibited rigorous communication-theoretic analysis of large-scale \acp{VN}.
In this work, we leverage a recently established closed-form analytical channel model for \acp{VN}, named \textit{\ac{MIGHT}}, to conduct the first systematic communication-theoretic study of \ac{MC} in complex, large-scale \acp{VN}. 
Based on \ac{MIGHT}, we derive a Poisson channel noise model and unveil structural analogies between \ac{MWC} and advective-diffusive \ac{MC} in \acp{VN}. In particular, we establish classical \ac{MWC} metrics, namely the \textit{\ac{RMS} delay spread}, the \textit{mean excess delay}, and the \textit{coherence bandwidth}, for \ac{MC} in \acp{VN} and derive closed-form expressions for the channel frequency response and \ac{PDP}.
Building on this characterization, we propose a \ac{VN}-adapted, coherent \ac{DF} detector and show how the derived multipath metrics can inform the choice of critical system parameters like the symbol duration, the sampling time, and the memory length. 
Additionally, we evaluate the detector's performance in different \acp{VN} exhibiting \ac{ISI}. 
Together, these contributions open the door to a systematic, \ac{MWC}-inspired \ac{MC} system design for large-scale \acp{VN}.
\end{abstract}

\acresetall


\section{Introduction}\label{sec:Introduction}
In \ac{MC}, information is encoded into the properties of signaling molecules. 
Although the field is inherently interdisciplinary, its early development was strongly influenced by conventional communications engineering. 
Foundational concepts from wireless communications, including channel modeling, modulation, equalization, detection, and channel estimation, were therefore initially transferred largely unchanged to \ac{MC}~\cite{Jamali2019, Kuran2021}.
This influence was also reflected in early \ac{MC} models for the human \ac{CVS}, which soon emerged as a promising application domain~\cite{Nakano2012}.
Since then, analytical channel models for \ac{MC} in the \ac{CVS} evolved from unbounded channels with flow and diffusion~\cite{Srinivas2012}, to single-duct channels, to branched tree structures~\cite{Chahibi2013}, and ultimately to complex multipath \ac{VN} topologies~\cite{Jakumeit2025b}.
While early channel models for simple environments such as single ducts enabled rigorous communication-theoretic analysis due to their analytical tractability, the increased structural detail required to capture the properties of complex \acp{VN} introduced substantial complexity, thereby limiting the feasibility of systematic communication-theoretic treatment.

Nonetheless, several studies identify conceptual parallels between branched \ac{MC} channels and wireless multipath channels. 
In~\cite{Bicen2013}, microfluidic \ac{MC} channels are modeled via transfer functions, inspired by classical signal processing. Although this establishes a formal analogy to communication-theoretic modeling, the considered channel topologies are structurally simple while the analysis remains comparatively involved.
In~\cite{Barros2013}, calcium signaling between neighboring cells is analyzed from a communications-theoretic perspective, revealing parallels to wireless multipath propagation due to multiple diffusive pathways and path-dependent delays. However, the study relies on numerical simulations and lacks a rigorous analytical evaluation of multipath propagation.
In~\cite{Wang2020}, a two-path fluidic \ac{MC} testbed is characterized using \textit{delay spread} and \textit{coherence time} metrics from \ac{MWC}. However, the reported delay spread is attributed primarily to diffusion rather than multipath propagation, and the analysis is limited to a simple two-path system without treating more complex \ac{VN} topologies.
In~\cite{Jakumeit2025b}, a convolution-based channel model for \ac{MC} in \acp{VN} is proposed and experimentally validated. 
To assess the impact of network topology, and thus, multipath propagation, on the received \ac{SNR}, two topology-dependent dispersion metrics are introduced. 
However, due to model complexity, these metrics cannot be directly mapped to established \ac{MWC} metrics.
Collectively, these works highlight parallels between \ac{MC} and \ac{MWC}, yet a rigorous, analytically tractable framework for characterizing multipath communication in complex \acp{VN} remains absent.

To address this gap, in this work, we focus on the communication-theoretic analysis and system design implications of advective-diffusive multipath propagation in \acp{VN}.
In particular, we leverage tools from \ac{MWC}, motivated by the following parallels between \ac{MWC} and \ac{MC} in \acp{VN}, see Fig.~\ref{fig:structural_parallels}:
\textbf{1)}~Signals propagate along multiple paths from the \ac{Tx} to the \ac{Rx}. In \ac{MWC}, multipath propagation arises from reflections and scattering of electromagnetic waves, whereas in \acp{VN} it results from branching and reconverging vasculature.
\textbf{2)}~Individual paths exhibit distinct delays and attenuations, leading to signal superposition at the \ac{Rx}.
\textbf{3)}~System design relies on the temporal and spectral characterization of the channel in order to assess dispersion, \ac{ISI}, and achievable communication performance.

In this paper, we base our analysis on the existing \textit{\ac{MIGHT}} model~\cite{Jakumeit2026}, which offers a convenient closed-form description for advective-diffusive molecule transport in complex \acp{VN}, thereby overcoming the analytical unwieldiness of existing \ac{VN} models.
\begin{figure}
    \centering
    \includegraphics[width=\linewidth]{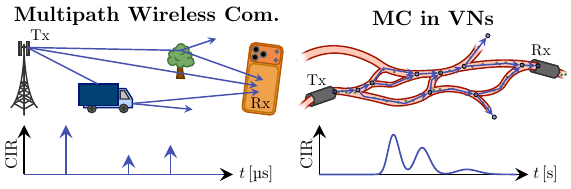}
    \caption{Structural parallels between \ac{MWC} and \ac{MC} in \acp{VN}.}
    \label{fig:structural_parallels}
\end{figure}
The contributions of this work are as follows:
\begin{itemize}
\itemsep0em
    \item Building on the deterministic \ac{MIGHT} channel model~\cite{Jakumeit2026}, we formulate a Poisson noise model for \ac{MC} in \acp{VN}.
    \item Leveraging structural parallels to \ac{MWC}, we transfer and adapt classical multipath channel metrics to advective-diffusive \ac{MC} in \acp{VN}, namely the \textit{\ac{RMS} delay spread}, the \textit{mean excess delay}, and the \textit{coherence bandwidth}.
    \item We obtain, for the first time in the \ac{MC} literature, a closed-form expression for the frequency response of advection-diffusion-driven \ac{VN} channels and analyze its magnitude, phase, and group delay across \acp{VN} of varying topological complexity, yielding direct insight into the frequency-selectivity of \ac{VN} channels without requiring numerical simulation.
    \item We develop a \ac{VN}-adapted sampling and coherent detection strategy that explicitly accounts for multipath-induced dispersion and analyze the dependence of \ac{ISI} and communication performance on the channel characteristics.
\end{itemize}

The remainder of the paper is organized as follows: 
Section~\ref{sec:System_Model} introduces the formal definition of \acp{VN}, assumptions on advective-diffusive molecule transport, the \ac{Tx} and \ac{Rx} models, the deterministic \ac{MIGHT} channel model, and a channel noise model.
This is followed by the derivation of metrics for the characterization of multipath channels in Section~\ref{sec:Multi-Path_Metrics}.
Section~\ref{sec:Coherent_Detectors} proposes sampling and detection strategies for \acp{VN}, and Section~\ref{sec:NumericalResults} presents the evaluation of the channel metrics and detection schemes.
Section~\ref{sec:Conclusion} concludes the paper. 

\section{System and Channel Model}\label{sec:System_Model}
\subsection{Vessel Network Definition}

We adopt the \ac{VN} definition from \cite{Jakumeit2026} and approximate any \ac{VN} topology using three segment types:
\begin{enumerate}
    \itemsep0em
    \item \textit{Pipe:} A pipe $p_i$ is a cylindrical conduit with length $l_i$ and radius $r_i$ transporting fluid from its inlet to its outlet. Pipes connect to other pipes, bifurcations, or junctions. The number of pipes in a \ac{VN} is denoted by $E$.
    \item \textit{Bifurcation:} A bifurcation $b_m$ is a \ac{0-D} connection, where one or more inflow pipe(s) split(s) into multiple outflow pipes. We denote the set of its outflow pipes by $\mathcal{O}(b_m)$. The number of bifurcations in a \ac{VN} is denoted by $B$.
    \item \textit{Junction:} A junction is a \ac{0-D} connection, where multiple inflow pipes merge into one outflow pipe.
\end{enumerate}
Bifurcations, junctions, inlet(s), outlet(s), and any connection point are modeled as nodes. We differentiate three node types:
\begin{enumerate}
\itemsep0em
    \item \textit{Inlet node:} Inlet nodes $n_{\mathrm{in},i}\in \mathcal{N}_\mathrm{in}=\{n_{\mathrm{in},1},\ldots ,n_{\mathrm{in},I}\}$ with $\,I \in \mathbb{N}$ exist at the points of the \ac{VN} where fluid flow is introduced into the \ac{VN}. Here, $\mathbb{N}$ denotes the set of natural numbers.
    \item \textit{Outlet node:} Outlet nodes $n_{\mathrm{out},i}\in \mathcal{N}_\mathrm{out}=\{n_{\mathrm{out},1},\ldots ,n_{\mathrm{out},O}\}$ with $\,O \in \mathbb{N}$ exist at the points of the \ac{VN} where fluid flow leaves the \ac{VN}.
    \item \textit{Connecting node:} All other $C\in\mathbb{N}$ points in the \ac{VN} where pipes are connected to one another are referred to as connecting nodes $n_i\in\mathcal{N}_\mathrm{con}=\{n_{1},\ldots ,n_{C}\}$. 
\end{enumerate}
Pipes are directed edges between nodes, aligned with the fluid flow direction, see Section~\ref{ssec:advective_diffusive_molecule_transport}.
For any node type, the nodes at the inlet and outlet of a pipe $p_i$, i.e., its \textit{source node} and \textit{destination node}, are denoted by $\mathcal{S}(p_i)$ and $\mathcal{D}(p_i)$, respectively.

The representation based on nodes and directed edges allows any \ac{VN} to be described as a directed multigraph, see right-hand side of Fig.~\ref{fig:structural_parallels}.
The set of all distinct directed paths between a given (inlet/connecting) node $n_{a}\in\mathcal{N}_\mathrm{in}\cup \mathcal{N}_\mathrm{con}$ and another (connecting/outlet) node $n_{b}\in\mathcal{N}_\mathrm{con}\cup \mathcal{N}_\mathrm{out}$ is denoted by $\mathcal{P}(n_{a}, n_{b})$.
Each path $P_g$ comprises a subset of pipes and bifurcations given by
\begin{equation}\label{eqn:PathSet}
P_g = \left\lbrace p_i \mid i \in \mathcal{E}_g \right\rbrace \cup \left\lbrace b_m \mid m \in \mathcal{B}_g \right\rbrace,
\end{equation}
where $\mathcal{E}_g \subseteq \left\lbrace 1,\ldots ,E \right\rbrace$ and $\mathcal{B}_g \subseteq \left\lbrace 1,\ldots ,B \right\rbrace$ are the index sets of the pipes and bifurcations\footnote{Note that potential bifurcations at $n_a$ or $n_b$ are \textit{not} included in the path set in~\eqref{eqn:PathSet}. This is because molecules in any path between $n_a$ and $n_b$ do not actually travel \textit{through} $n_a$ or $n_b$, but rather start at $n_a$ and end at $n_b$.} included in $P_g$. 

\subsection{Transmitter, Receiver, and Molecule Transport}\label{ssec:advective_diffusive_molecule_transport}
In this work, we focus on a single-\ac{Tx}, single-\ac{Rx} scenario.
The \ac{Tx} is positioned at $z_a = z_\mathrm{Tx}$ in pipe $p_a$, where $z_a\in[0,l_a]$ denotes the longitudinal coordinate within $p_a$. At time $t=0$, the \ac{Tx} impulsively\footnote{More sophisticated pulse shaping is left for future work; the impulsive 
injection allows the characterization of the channel \ac{ISI}.} releases $N$ signaling molecules.
Furthermore, we assume a transparent counting \ac{Rx}, centered around $z_b=z_\mathrm{Rx}\in[0,l_b]$ in pipe $p_b$ with length $l_\mathrm{Rx}$.
Moreover, 
\begin{equation}
    \mathcal{P}_{\mathrm{Tx},\mathrm{Rx}}=\{P_g\in \mathcal{P}(\mathcal{S}(p_a),\mathcal{D}(p_b))\,\vert\, p_a,p_b\in P_g\}
\end{equation}
contains all paths from the \ac{Tx} to the \ac{Rx} that include $p_a$ and~$p_b$.

Molecules propagate driven by advection and diffusion and under assumption of the Aris-Taylor regime, see~\cite{Jakumeit2026}.
At the inlet nodes $n_{\mathrm{in},i}$, time-invariant\footnote{Time-invariant flow is a valid approximation in small/medium-sized vessels where pulsatility is damped by the Windkessel effect of preceding large arteries.} flow rates $Q_{\mathrm{in},i}$ are applied; 
at the outlet nodes, zero hydraulic resistance is assumed. 
The flow velocities $\bar{u}_i$ and flow rates $Q_i$ are obtained from an equivalent  electrical circuit model~\cite{Jakumeit2026}, and diffusive transport in pipe $p_i$ is governed by the effective diffusion coefficient $\bar{D}_i = r_i^2\bar{u}_i^2/(48D) + D$~\cite[Eq.~(26)]{Aris1956}, where $D$ denotes the molecular diffusion coefficient.

\subsection{Channel Impulse Response}
Employing the \ac{MIGHT} model~\cite{Jakumeit2026}, the molecule flux $\bar{j}_g$ in $\SI{}{\per\second}$ through path $P_g$, observed at position $z_{\mathrm{Rx}}$ in pipe $p_b$ due to the release of $N=1$ molecule at the \ac{Tx} is given as~\cite[Eqs.~(8), (14)]{Jakumeit2026}
\begin{multline}\label{eqn:path_flux}
    \bar{j}_g(z_\mathrm{Rx},t;z_\mathrm{Tx})=\\\frac{\bar{\mu}_g(z_\mathrm{Rx};z_\mathrm{Tx})}{\sqrt{2\pi\bar{\theta}_g(z_\mathrm{Rx};z_\mathrm{Tx})t^3}}\exp\left(-\frac{\left(t-\bar{\mu}_g(z_\mathrm{Rx};z_\mathrm{Tx})\right)^2}{2\bar{\theta}_g(z_\mathrm{Rx};z_\mathrm{Tx})t}\right),
\end{multline}
with the mean, variance, and scale parameter of path $P_g$ given by~\cite[Eq.~(15), (16), (17)]{Jakumeit2026}
\begin{align}
        &\bar{\mu}_g(z_\mathrm{Rx};z_{\mathrm{Tx}})\hspace*{-.5mm}=\hspace*{-.5mm}
        \mu_a(l_a\hspace*{-.5mm}-\hspace*{-.5mm}z_{\mathrm{Tx}})\hspace*{-.5mm}+\hspace*{-.5mm}\mu_b(z_\mathrm{Rx})\hspace*{-.5mm}+\hspace*{-.5mm}\sum_{i\in \mathcal{E}_g\backslash \left\lbrace a,b\right\rbrace}\mu_i(l_i),\label{eq:mom-het-mu}\\
         &\bar{\sigma}_g^2(z_\mathrm{Rx};z_{\mathrm{Tx}})
        \hspace*{-.5mm}=\hspace*{-.5mm} \sigma_a^2(l_a\hspace*{-.5mm}-\hspace*{-.5mm}z_{\mathrm{Tx}})\hspace*{-.5mm}+\hspace*{-.5mm}\sigma_b^2(z_\mathrm{Rx})\hspace*{-.5mm}+\hspace*{-2mm}\sum_{i\in \mathcal{E}_g\backslash \left\lbrace a,b\right\rbrace}\hspace*{-1mm}\sigma_i^2(l_i),\label{eq:mom-het-sigma}\\
        &\bar{\theta}_g(z_\mathrm{Rx};z_{\mathrm{Tx}})\hspace*{-.5mm}=\hspace*{-.5mm}\frac{\bar{\sigma}_g^2(z_\mathrm{Rx};z_{\mathrm{Tx}})}{\bar{\mu}_g(z_\mathrm{Rx};z_{\mathrm{Tx}})},\label{eq:mom-het-theta}
\end{align}
and the mean and variance of pipe $p_i$ given by~\cite[Eq.~(7)]{Jakumeit2026}
\begin{align}
    &\mu_i(z_i) = \frac{z_i}{\bar{u}_i}, 
    &\sigma_i^2(z_i) = \frac{2\bar{D}_i z_i}{\bar{u}_i^3}\,.
    \label{eqn:pipe_mean_variance_scale}
\end{align}
The expression in~\eqref{eqn:path_flux} can also be interpreted as the \ac{PDF} of the path \ac{FPT} $T_g$~\cite{Jakumeit2026}.
Assuming dominant advective flux compared to diffusive flux,
the \ac{CIR} between the \ac{Tx} and the transparent \ac{Rx} in $\SI{}{\per\meter}$ is then obtained under the \ac{UCA} as~\cite[Eq.~(24)]{Jakumeit2026}
\begin{equation}\label{eqn:CIR_between_Tx_and_RX}
    h(z_\mathrm{Rx},t;z_{\mathrm{Tx}})= \frac{l_\mathrm{Rx}}{\bar{u}_b}\sum_{P_g\in\mathcal{P}_{\mathrm{Tx},\mathrm{Rx}}}\gamma_{P_g} \bar{j}_g(z_\mathrm{Rx},t;z_{\mathrm{Tx}})\,.
\end{equation}
Here, $\gamma_{P_g}\in[0,1]$ denotes the fraction of molecules propagating through path $P_g$, which is obtained from the fractions of flow rates at each bifurcation $b_m\in P_g$ as~\cite[Eq.~(22)]{Mosayebi2019}  
\begin{equation}\label{eqn:path_fraction}
    \gamma_{P_g} = \prod\limits_{\substack{p_i,b_m\in P_g,\\ p_i\in\mathcal{O}(b_m )}} \dfrac{Q_{i}}{\sum_{p_v\in\mathcal{O}(b_m)} Q_{v}}\,.
\end{equation}
In the remainder of this paper, for brevity, we drop the spatial arguments $z_\mathrm{Tx}$ and $z_\mathrm{Rx}$ in the notation in~\eqref{eqn:path_flux}--\eqref{eq:mom-het-theta}, \eqref{eqn:CIR_between_Tx_and_RX}, and denote the \ac{CIR} as $h(t)$, the path parameters as $\bar{\mu}_g$, $\bar{\sigma}_g^2$, $\bar{\theta}_g$, and the path flux as $\bar{j}_g(t)$. 

\subsection{Discrete Noise Model}\label{ssec:discrete_time_noise_model}

We assume the \ac{Rx} takes one sample at time $t_\mathrm{s}$ per sent symbol $s[k]\in\{0,1\}$, $k\in\mathbb{Z}$, where \ac{OOK} is chosen as modulation scheme and $\mathbb{Z}$ is the set of integers. 
Moreover, we denote the \textit{expected} number of background noise molecules observed at the \ac{Rx} at any time by $\bar{n}$, which can be attributed, e.g., to biomarkers present in the \ac{CVS} that the \ac{Rx} cannot distinguish from the signaling molecules.
By employing a Poisson counting noise model (please refer to the appendix for a detailed derivation), the sample collected in the $k$-th symbol interval of duration $T_\mathrm{s}$ of the noisy received signal is given as
\begin{equation}\label{eqn:discrete_noise_model}
    r[k]\sim\mathrm{Pois}\left( \sum_{l=0}^{L-1} d[l]s[k-l]+\bar{n}\right),
\end{equation}
where
\begin{equation}
    d[l]=Nh(lT_\mathrm{s}+t_\mathrm{s})
\end{equation}
denotes the expected number of molecules observed at the \ac{Rx} due to the release of $N$ molecules at the \ac{Tx} at the beginning of the $(k-l)$-th symbol interval. Furthermore, $\mathrm{Pois}(\lambda)$ denotes a Poisson-distributed \ac{RV} with mean $\lambda$ and $L$ denotes the considered channel memory length. 

\section{Multipath Channel Metrics}\label{sec:Multi-Path_Metrics}
In this section, in analogy to \ac{MWC}, we derive the established multipath channel metrics \textit{\ac{RMS} delay spread} and \textit{coherence bandwidth} for the case of \ac{MC} in \acp{VN} and present a closed-form expression for the frequency response of the \ac{VN} channel.

\subsection{Root Mean Squared Delay Spread and Mean Excess Delay}\label{ssec:RMS_delay_spread}
In \ac{MWC}, the \ac{RMS} delay spread $\tau_\mathrm{RMS}$ is a measure for the time dispersion that the channel induces.
It is calculated from the \textit{\ac{PDP}}, which in turn is based on the squared \ac{CIR}~\cite{Rappaport2002}. 
Below, we derive the \ac{PDP}, the \textit{mean excess delay}, and the \ac{RMS} delay spread for \ac{MC} in \acp{VN}. 

\begin{theorem}
    The \ac{PDP} for \ac{MC} in \acp{VN} is given by
    \begin{equation}\label{eqn:PDP}
    f_T(t)=\sum_{P_g\in\mathcal{P}_{\mathrm{Tx},\mathrm{Rx}}} w_g \bar{j}_g(t),\quad\text{with}\quad\int_0^\infty f_T(t)\,\mathrm{d}t=1,
    \end{equation}
    where $f_T(t)$ denotes the \ac{PDF} of the network \ac{FPT} $T$, i.e., the delay experienced by a molecule propagating from the \ac{Tx} to the \ac{Rx} in a \ac{VN}, and $w_g$ is the normalized weight of path $P_g$. 
\end{theorem}
\begin{proof}
    We first obtain the probability that a molecule released at the \ac{Tx} reaches the \ac{Rx} as\footnote{Note that $\chi<1$ if some paths from the \ac{Tx} never reach the \ac{Rx}, see Fig.~\ref{fig:structural_parallels}.}
    \begin{equation}
    \chi=\sum_{P_g\in\mathcal{P}_{\mathrm{Tx},\mathrm{Rx}}}\gamma_{P_g}\,.
    \end{equation}
    Accounting for $\chi$, the normalized weight of path $P_g$ is
    \begin{equation}\label{eqn:normalized_weights}
    w_g=\frac{\gamma_{P_g}}{\chi},\quad\text{with}\quad \sum_{P_g\in\mathcal{P}_{\mathrm{Tx},\mathrm{Rx}}}w_g=1\,.
    \end{equation}
    The \ac{PDP} in~\eqref{eqn:PDP} follows because $\bar{j}_g(t)$ is a \ac{PDF} over~$t$.
\end{proof}
In contrast to the classical \ac{PDP} definition in \ac{MWC} based on the squared \ac{CIR}~\cite{Rappaport2002}, we formulate $f_T(t)$ in terms of molecule flux, consistent with its interpretation as a \ac{PDF} over $t$ and ensuring that the moments of $T$ have direct physical meaning.
Although~\eqref{eqn:PDP} is not directly related to power, we choose to preserve the term "\ac{PDP}" to emphasize the analogy to \ac{MWC}.

\begin{theorem}
    The mean excess delay $\mathbb{E}[T]$ and the \ac{RMS} delay spread $\tau_\mathrm{RMS}$ for \ac{MC} in \acp{VN}, defined as the first moment and standard deviation of $T$~\cite{Rappaport2002}, respectively, are obtained as
    \begin{align}
        \mathbb{E}[T]&=\sum_{P_g\in\mathcal{P}_{\mathrm{Tx},\mathrm{Rx}}}w_g\bar{\mu}_g,\\
        \tau_\mathrm{RMS}\hspace*{-.5mm}&=\hspace*{-2.5mm}\sqrt{\underbrace{\sum_{P_g\in\mathcal{P}_{\mathrm{Tx},\mathrm{Rx}}} \hspace*{-5mm}w_g\bar{\mu}_g\bar{\theta}_g}_{\text{Diffusion spread}} + \hspace*{-2mm}\underbrace{\sum_{P_g\in\mathcal{P}_{\mathrm{Tx},\mathrm{Rx}}} \hspace*{-5mm}w_g\bar{\mu}_g^2 -\hspace*{-1mm}\left(\hspace*{-.5mm}\sum_{P_g\in\mathcal{P}_{\mathrm{Tx},\mathrm{Rx}}} \hspace*{-5mm}w_g\bar{\mu}_g\hspace*{-1mm} \right)^2}_{\text{Multipath spread, variance of }\bar{\mu}_g}},\label{eqn:RMS_delay_spread}
    \end{align}
    where $\mathbb{E}[\cdot]$ denotes the expectation operator.
\end{theorem}
In~\eqref{eqn:RMS_delay_spread}, the first term under the root captures diffusion-induced dispersion, the second multipath-induced dispersion.
\begin{proof}
    For each path $P_g$, the first, second central, and second moments of $T_g$ are obtained as
\begin{align}\label{eqn:path_delay_moments}
    &\mathbb{E}\left[T_g\right]\hspace*{-.5mm}\overset{\eqref{eq:mom-het-mu}}{=}\hspace*{-.5mm}\bar{\mu}_g,\hspace*{-2.6mm}
    & \mathrm{Var}(T_g)\hspace*{-.5mm}\overset{\eqref{eq:mom-het-sigma}}{=}\hspace*{-.5mm}\bar{\sigma}^2_g\hspace*{-.5mm}\overset{\eqref{eq:mom-het-theta}}{=}\hspace*{-.5mm}\bar{\mu}_g\bar{\theta}_g,\hspace*{-2.6mm}
    && \mathbb{E}\left[T_g^2\right]\hspace*{-.5mm}=\hspace*{-.5mm}\bar{\mu}_g^2\hspace*{-.5mm}+\hspace*{-.5mm}\bar{\mu}_g\bar{\theta}_g,
\end{align}
where $\mathrm{Var}(\cdot)$ denotes the variance operator.
From~\eqref{eqn:path_delay_moments}, the first and second moments of $T$ are derived as
\begin{align}
    &\mathbb{E}\left[T\right]=\int_0^\infty t\, f_T(t)\, \mathrm{d}t \overset{\eqref{eqn:PDP}}{=} \int_0^\infty t \sum_{P_g\in\mathcal{P}_{\mathrm{Tx},\mathrm{Rx}}} \hspace*{-3mm}
w_g \bar{j}_g(t)\, \mathrm{d}t \nonumber\\
&\overset{\eqref{eqn:normalized_weights}}{=} \hspace*{-3mm}\sum_{P_g\in\mathcal{P}_{\mathrm{Tx},\mathrm{Rx}}} \hspace*{-4mm}w_g 
\int_0^\infty \hspace*{-2mm}t\, \bar{j}_g(t)\, \mathrm{d}t = \hspace*{-4mm}\sum_{P_g\in\mathcal{P}_{\mathrm{Tx},\mathrm{Rx}}}\hspace*{-4mm} w_g \mathbb{E}[T_g]\overset{\eqref{eqn:path_delay_moments}}{=}\hspace*{-4mm}\sum_{P_g\in\mathcal{P}_{\mathrm{Tx},\mathrm{Rx}}}\hspace*{-4mm} w_g \bar{\mu}_g,\label{eqn:first_moment_T}\\
    &\mathbb{E}\left[T^2\right]=\hspace*{-1mm}\sum_{P_g\in\mathcal{P}_{\mathrm{Tx},\mathrm{Rx}}} \hspace*{-3mm}w_g\mathbb{E}\left[T_g^2\right]\overset{\eqref{eqn:path_delay_moments}}{=}\hspace*{-2mm}\sum_{P_g\in\mathcal{P}_{\mathrm{Tx},\mathrm{Rx}}} \hspace*{-3mm}w_g\left(\bar{\mu}_g^2+\bar{\mu}_g\bar{\theta}_g\right)\,.\label{eqn:second_moment_T}
\end{align}
Lastly, the \ac{RMS} delay spread of the channel between the \ac{Tx} and \ac{Rx} follows as the standard deviation of $T$
\begin{equation}\label{eqn:RMS_delay_spread_formula}
    \tau_\mathrm{RMS}=\sqrt{\mathrm{Var}(T)}=\sqrt{\mathbb{E}\left[T^2\right] - \left(\mathbb{E}\left[T\right]\right)^2}\,.
\end{equation}
Inserting~\eqref{eqn:first_moment_T} and~\eqref{eqn:second_moment_T} into~\eqref{eqn:RMS_delay_spread_formula} yields~\eqref{eqn:RMS_delay_spread}.
\end{proof}

\subsection{Coherence Bandwidth}

In \ac{MWC}, the \textit{coherence bandwidth} $B_\mathrm{c}$ is the bandwidth over which the channel's frequency response $H(f)$ remains highly correlated across frequency. Analogous to \ac{MWC}, for \acp{VN}, the coherence bandwidth in $\SI{}{\hertz}$ can be roughly estimated from the \ac{RMS} delay spread as~\cite[Eq.~(12)]{Sklar1997}
\begin{equation}\label{eqn:coherence_bandwidth}
    B_\mathrm{c} \approx \frac{1}{2\pi \tau_\mathrm{RMS}}\,.
\end{equation}
A small $B_\mathrm{c}\hspace*{-.5mm}<\hspace*{-.5mm}\frac{1}{T_\mathrm{s}}$ indicates a highly frequency-selective~channel.

\subsection{Frequency Response, Magnitude, Phase, and Group Delay}

It is important to note that the expression in~\eqref{eqn:coherence_bandwidth} only yields a rough approximation of $B_\mathrm{c}$. 
A full spectral analysis 
of the channel provides a more complete picture of its frequency-selectivity.
To this end, we derive the frequency response $H(f)$ of the \ac{VN} channel from the \ac{CIR} in \eqref{eqn:CIR_between_Tx_and_RX} as
\begin{align}\label{eqn:frequency_response}
    &H(f)=\mathcal{F}\{h(t)\}=\int_0^\infty h(t)\exp (-\mathrm{j}2\pi ft)\,\mathrm{d}t\nonumber\\
    &=\frac{l_\mathrm{Rx}}{\bar{u}_b}\sum_{P_g\in\mathcal{P}_{\mathrm{Tx},\mathrm{Rx}}}\hspace*{-3mm}\gamma_{P_g} \exp\left(\frac{\bar{\mu}_g}{\bar{\theta}_g}\left(1-\sqrt{1+\mathrm{j}4\pi \bar{\theta}_g f}\right)\right),
\end{align}
where $f$, $\mathcal{F}\{\cdot\}$, and $\mathrm{j}$ denote the frequency, the Fourier transform with respect to $t$, and the imaginary unit, respectively. In Section~\ref{sec:NumericalResults}, we investigate the magnitude $|H(f)|$, phase $\phi (f)=\arg \{H(f)\}$, and group delay $\tau_\mathrm{g}(f)=-\frac{1}{2\pi}\frac{\mathrm{d}\phi(f)}{\mathrm{d}f}$, as they jointly provide a comprehensive picture of the channel's characteristics in the frequency domain.

\section{Sampling and Coherent Detection}\label{sec:Coherent_Detectors}
Below, we first propose a simple \ac{VN}-adapted sampling strategy.
Second, we detail how the minimum symbol duration required to avoid \ac{ISI} can be estimated from the \ac{RMS} delay spread.
Third, we derive a \ac{DF} detector employing the proposed sampling strategy.

\subsection{Strongest Path Peak Time Sampling}
We propose the \textit{strongest path peak time sampling} strategy, which, 
within one symbol interval of length $T_\mathrm{s}$, samples the received 
signal at the peak time $t^\mathrm{peak}_{g^*}$ of the path $P_{g^*}$ whose weighted 
flux contribution $\gamma_{P_{g^*}}\bar{j}_{g^*}(t)$ attains the highest peak value. 
The sampling time follows as
\begin{equation}
    t_\mathrm{s} = t^\mathrm{peak}_{g^*},\quad\text{with}\quad 
    g^*=\argmax_g \, \gamma_{P_g}\bar{j}_g\!\left(t_g^\mathrm{peak}\right)\,.
\end{equation}
The peak time $t^\mathrm{peak}_{g}$ of path $P_g$ is obtained from the path 
flux in~\eqref{eqn:path_flux} by solving
\begin{equation}\label{eqn:path_peak_time_condition}
    \frac{\partial }{\partial t}\bar{j}_g(t)\hspace*{-.5mm} 
    \overset{\eqref{eqn:path_flux}}{=} \hspace*{-.7mm}-\frac{\bar{\mu}_g 
    t\left(t^2 \hspace*{-.5mm}+\hspace*{-.5mm} 3\bar{\theta}_g t 
    \hspace*{-.5mm}-\hspace*{-.5mm} \bar{\mu}_g^2 \right)}{\sqrt{\pi}
    \left(2\bar{\theta}_g t^3\right)^{3/2}}\hspace*{-.5mm}\exp\hspace*{-.5mm}
    \left(\hspace*{-.7mm}-\frac{\left(t\hspace*{-.5mm}-\hspace*{-.5mm}
    \bar{\mu}_g\right)^2}{2\bar{\theta}_gt}\hspace*{-.5mm}\right)
    \hspace*{-1mm}=\hspace*{-.5mm}0.
\end{equation}
Keeping only the solution that satisfies $t>0$ yields
\begin{align}\label{eqn:candidate_peak_times}
t_g^\mathrm{peak}=\frac{-3\bar{\theta}_g+\sqrt{9\bar{\theta}_g^2+
4\bar{\mu}_g^2}}{2}\,.
\end{align}
Here, $\lim_{\bar{\theta}_g\to 0} t_g^\mathrm{peak} 
\overset{\eqref{eqn:candidate_peak_times}}{=} \bar{\mu}_g$, i.e., for 
limited diffusion, the path peak time approaches the path mean $\bar{\mu}_g$.

\subsection{Inter-Symbol Interference-Dependent Symbol Duration}

In diffusion-only and single-duct advective-diffusive \ac{MC} systems, \ac{ISI} arises from long tails of the \ac{CIR} caused by molecular diffusion and effective dispersion under laminar flow, respectively. In complex \acp{VN}, topology-induced multipath propagation further increases signal spread, as evident from the \ac{RMS} delay spread in~\eqref{eqn:RMS_delay_spread}, thereby aggravating \ac{ISI}.

In \ac{MWC}, the \ac{RMS} delay spread is a key parameter for \ac{ISI} mitigation, and the symbol duration is typically chosen proportionally to it, as most signal energy is concentrated within approximately $\mathbb{E}[T]\pm 2\tau_\mathrm{RMS}$. Adopting this established rule, we select~\cite[Eq.~(4.42)]{Rappaport2002}
\begin{equation}\label{eqn:symbol_duration_rule}
    T_\mathrm{s} \geq c \cdot \tau_\mathrm{RMS},
\end{equation}
where $c$ typically lies between $4$ and $6$, depending on the tolerable \ac{ISI} level. The applicability of this design rule to \ac{MC} in \acp{VN} is evaluated in Section~\ref{sec:NumericalResults}.

\begin{figure*}
    \centering
    \includegraphics[width=\linewidth]{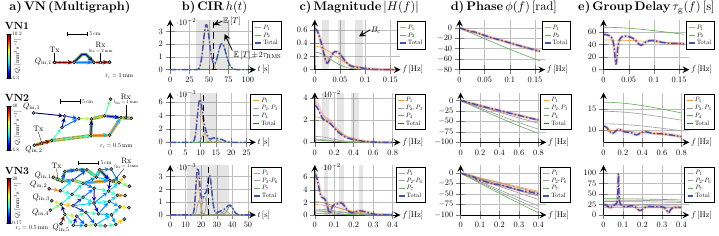}
    \caption{Time- and frequency-domain characterization of \acp{VN}. \textbf{a)}~Vessel lengths are drawn to scale, radii are given below each \ac{VN}. Flow rates and fastest (orange)/slowest (green) paths are color-coded. \ac{Tx} and \ac{Rx} positions are marked. Inlet flow rates in $\SI{}{\meter\cubed\per\second}$ are $Q_{\mathrm{in},1}=1.021\hspace*{-.5mm}\times\hspace*{-.5mm} 10^{-8}$ for \ac{VN}1, $Q_{\mathrm{in},\{ 1,2\}}=\{ 0.5,4 \}\hspace*{-.5mm}\times\hspace*{-.5mm} 10^{-8}$ for \ac{VN}2, and $Q_{\mathrm{in},\{ 1,\ldots,5\}}=\{ 1,1.5,1.2,1,2 \}\hspace*{-.5mm}\times\hspace*{-.5mm} 10^{-8}$ for \ac{VN}3, and $D=\SI{1.46e-7}{\meter\squared\per\second}$. \textbf{b)}~\acp{CIR} and path contributions are shown. Mean excess delay $\mathbb{E}[T]$ and the interval $\mathbb{E}[T]\pm2\tau_\mathrm{RMS}$ containing most energy are marked by a black dashed line and gray background, respectively. \textbf{c)--e)}~Magnitude, phase, and group delay are shown. Thick red curves are numerically computed from the \acp{CIR}, thick blue curves are based on~\eqref{eqn:frequency_response}. In \textbf{c)}, the coherence bandwidth $B_\mathrm{c}$ is shown by a gray background at several frequencies.}
    \label{fig:spectra}
\end{figure*}

\subsection{Coherent Adaptive Decision-Feedback Detector}

Below, we briefly derive a coherent \ac{DF} detector for the noise model proposed in Section~\ref{ssec:discrete_time_noise_model}. First, in a genie-aided approach (i.e., assuming knowledge of $s[k-l],l\in\{1,\ldots , L-1\}$), we define the conditional Poisson rates given the \textit{true} transmit symbol $s[k]=b\in\{0,1\}$, i.e., 
\begin{align}
    \lambda_b&=b d[0] \hspace{-.5mm}+\hspace{-.5mm} \sum_{l=1}^{L-1} d[l]s[k-l]\hspace{-.5mm}+\hspace{-.5mm}\bar{n} \defeq b d[0] \hspace{-.5mm}+\hspace{-.5mm}\lambda_\mathrm{ISI}^+,\label{eqn:Poisson_rates}
\end{align}
where $\lambda_\mathrm{ISI}^+$ denotes the \ac{ISI} plus noise term. Using~\eqref{eqn:Poisson_rates}, the likelihood of observing $r[k]=r$ under hypothesis $s[k]=b$ is
\begin{equation}\label{eqn:likelihood}
    P(r[k]=r\,|\, s[k]=b) = \frac{\lambda_b^r}{r!}\exp (-\lambda_b)\,.
\end{equation}
The detector estimates the transmitted symbol $s[k]$ using
\begin{equation}
    \hat{s}[k]=\argmax_{b\,\in\{0,1\}} P(r[k]=r\,|\, s[k]=b),
\end{equation}
with $P(s[k]\hspace*{-.5mm}=\hspace*{-.5mm}0)\hspace*{-.5mm}=\hspace*{-.5mm}P(s[k]\hspace*{-.5mm}=\hspace*{-.5mm}1)\hspace*{-.5mm}=\hspace*{-.5mm}0.5$ assumed henceforth.
By deducing the log-likelihood from~\eqref{eqn:likelihood} and dropping the $\ln(r!)$ term, we conclude that the detector decides for $s[k]=1$ if
\begin{align}\label{eqn:ML_decision_rule}
    \ln\left(P(r[k]\hspace{-.5mm}=\hspace{-.5mm}r| s[k]\hspace{-.5mm}=\hspace{-.5mm}1) \right)&> \ln\left(P(r[k]\hspace{-.5mm}=\hspace{-.5mm}r| s[k]\hspace{-.5mm}=\hspace{-.5mm}0) \right)\nonumber\\
    \Leftrightarrow\quad r&>\frac{\lambda_1 - \lambda_0}{\ln\left(\frac{\lambda_1}{\lambda_0}\right)}\,.
\end{align}
By inserting $\lambda_1-\lambda_0=d[0]$ and $\lambda_1/\lambda_0=(d[0] + \lambda_\mathrm{ISI}^+)/\lambda_\mathrm{ISI}^+$, the decision rule in~\eqref{eqn:ML_decision_rule} simplifies to
\begin{equation}\label{eqn:final_decision_rule}
    r[k]\overunderset{\hat{s}[k]=1}{\hat{s}[k]=0}{\gtrless}\frac{d[0]}{\ln\left(1 + \frac{d[0]}{\lambda_\mathrm{ISI}^+}\right)}=\psi,
\end{equation}
where $\psi$ denotes the adaptive decision threshold, accounting for previously sent symbols.
In practice, since the \textit{true} transmitted symbols $s$ are not known, $\lambda_\mathrm{ISI}^+$ is estimated from the previously \textit{estimated} transmitted symbols $\hat{s}$ as
\begin{equation}
    \lambda_\mathrm{ISI}^+\approx \sum_{l=1}^{L-1} d[l]\hat{s}[k-l]\hspace{-.5mm}+\hspace{-.5mm}\bar{n},
\end{equation}
i.e., the rule in~\eqref{eqn:final_decision_rule} leads to an \textit{adaptive \ac{DF}}~architecture.

Lastly, we define the \ac{SER} as
\begin{align}\label{eqn:SER}
    \hspace*{-3mm}\mathrm{SER}\hspace*{-.6mm}=\hspace*{-.6mm}P(\hat{s}[k]\hspace*{-.6mm}\neq\hspace*{-.6mm} s[k])&\hspace*{-.6mm}=\hspace*{-.6mm}P(s[k]\hspace*{-.7mm}=\hspace*{-.7mm}0)P(\hat{s}[k]\hspace*{-.7mm}=\hspace*{-.7mm}1\vert s[k]\hspace*{-.7mm}=\hspace*{-.7mm}0) \nonumber\\
    &\hspace*{-.6mm}+\hspace*{-.6mm} P(s[k]\hspace*{-.7mm}=\hspace*{-.7mm}1)P(\hat{s}[k]\hspace*{-.7mm}=\hspace*{-.7mm}0\vert s[k]\hspace*{-.7mm}=\hspace*{-.7mm}1).
\end{align}

\section{Results}\label{sec:NumericalResults}
We illustrate the parallels between \ac{MWC} and \ac{MC} using three \acp{VN} of increasing structural complexity, see Fig.~\ref{fig:spectra}a), by first conducting a time- and frequency-analysis of the systems, and subsequently analyzing the communication performance in the \acp{VN} using the proposed detector.

\begin{figure*}
    \centering
    \includegraphics[width=\linewidth]{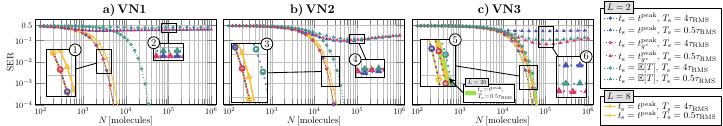}
    \caption{\acp{SER} of the proposed \ac{DF} detector for \textbf{a)} VN1, \textbf{b)} VN2, and \textbf{c)} VN3. Three sampling strategies and two symbol durations $T_\mathrm{s}=\{0.5\tau_\mathrm{RMS},4\tau_\mathrm{RMS}\}$ are considered, corresponding to high- and low-\ac{ISI} regimes, respectively. Moreover, channel memory $L$ is varied.}
    \label{fig:detector_evaluation}
\end{figure*}

\subsection{Time- and Frequency-Analysis of Vessel Network Channels}

Depending on the \ac{VN} topology and associated path means $\bar{\mu}_g$ and variances $\sigma^2_g$, the \acp{CIR} of \acp{VN} can be rather unimodal, see \ac{CIR} of VN2, or multimodal, see \acp{CIR} of VN1 and VN3.
In all cases, $\mathbb{E}[T]\pm 2\tau_\mathrm{RMS}$ reliably estimates the time interval containing most signal energy, see gray areas in Fig.~\ref{fig:spectra}b).
The gray areas appear to be slightly left-shifted, which stems from the general positive skew of \ac{VN} \acp{CIR} due to early paths carrying most of the molecules.

Considering $|H(f)|$, it can be seen that \acp{VN} generally behave like low-pass filters, since (effective) diffusion prohibits the transmission of high-frequency molecular signals. 
Cutoff frequencies as well as the coherence bandwidths $B_\mathrm{c}$, shown as gray areas in Fig.~\ref{fig:spectra}c), depend on the delay spread $\tau_\mathrm{RMS}$ of the \ac{VN}, see~\eqref{eqn:coherence_bandwidth}. For instance, VN1 exhibits a large $\tau_\mathrm{RMS}=\SI{10.94}{\second}$ and small $B_\mathrm{c}=\SI{0.015}{\hertz}$, whereas VN2 exhibits a shorter $\tau_\mathrm{RMS}=\SI{2.17}{\second}$ and consequently greater $B_\mathrm{c}=\SI{0.07}{\hertz}$.
The path magnitudes shown in Fig.~\ref{fig:spectra}c) do not necessarily add up to the total magnitude, because phase differences between complex-valued path contributions cause constructive and destructive interference, see~\eqref{eqn:frequency_response}, resulting in ripples in the total magnitude. 
However, it is visible that in the investigated \acp{VN}, the total magnitude generally follows the curve of the dominant path contribution (according to path weight $\gamma_g$, see~\eqref{eqn:frequency_response}). This is also true for the phase and group delay in Figs.~\ref{fig:spectra}d) and \ref{fig:spectra}e).
Local minima in $|H(f)|$ (and in $\phi(f)$ and $\tau_\mathrm{g}(f)$) occur at frequencies where destructive path interference is particularly strong. These frequencies are thus unsuitable for communication due to high frequency-selectivity.

Fig.~\ref{fig:spectra}d) shows the unwrapped phase contributions of different paths. The phases of paths with larger means $\bar{\mu}_g$ decrease faster with frequency than those with smaller means, see, e.g., $P_2$ and $P_1$ in VN1. Larger phase differences between paths lead to more frequent ripples in $|H(f)|$.

The total group delay at $f=0$ corresponds to the mean excess delay of the \ac{VN}, i.e., $\tau_\mathrm{g}(0)=\mathbb{E}[T]$. Comparing the \acp{CIR} and the group delay plots in Fig.~\ref{fig:spectra}e), it can be seen that individual path group delays at $f=0$ equal the path means~$\bar{\mu}_g$.
Note that $\tau_\mathrm{g}(f>0)$ does not admit a direct physical interpretation in terms of molecule transport in \acp{VN}, but is nonetheless shown in Fig.~\ref{fig:spectra}e) for completeness.

Overall, the time- and frequency-domain characterizations presented here provide a principled basis for \ac{MC} system design in \acp{VN}, see also Section~\ref{ssec:detector_evaluation}. 
The moments $\mathbb{E}[T]$ and $\tau_\mathrm{RMS}$ directly inform the choice of $T_\mathrm{s}$ via~\eqref{eqn:symbol_duration_rule}, while $B_\mathrm{c}$ provides a compact indicator of the channel's frequency-selectivity.

\subsection{Evaluation of Vessel Network-Adapted Detector}\label{ssec:detector_evaluation}

For each of the exemplary \acp{VN} in Fig.~\ref{fig:spectra}a), we evaluate the \ac{SER} in~\eqref{eqn:SER} as a function of the number of released molecules $N\in[10^2,10^6]$, see Fig.~\ref{fig:detector_evaluation}.
We focus on a low-\ac{ISI} scenario with long symbol duration $T_\mathrm{s}=4\tau_\mathrm{RMS}$ and a high-\ac{ISI} scenario with very short symbol duration $T_\mathrm{s}=0.5\tau_\mathrm{RMS}$.
For each \ac{ISI} regime, we test three different sampling strategies: \textbf{1)} Sampling at the numerically determined global peak time $t_\mathrm{s}=t^\mathrm{peak}$ of the \ac{CIR} in~\eqref{eqn:CIR_between_Tx_and_RX}, which maximizes the direct tap $d[0]$ and is thus optimal in the absence of \ac{ISI}, \textbf{2)} sampling at the strongest path peak time $t_\mathrm{s}=t_{g^*}^\mathrm{peak}$, given in~\eqref{eqn:candidate_peak_times}, which, in contrast to $t^\mathrm{peak}$, can be derived analytically and purely from path parameters (without requiring full \ac{CIR} superposition), and \textbf{3)}~sampling at the mean excess delay $t_\mathrm{s}=\mathbb{E}[T]$, given in~\eqref{eqn:path_delay_moments}, which represents the average molecule arrival time and thus intuitively places the sampling instant at the temporal center of the received signal.
The detector has a comparatively short default channel memory of $L=2$, i.e., one previous symbol estimate is used to estimate the \ac{ISI}, and the expected number of noise molecules is $\bar{n}=500$.
Additionally, the yellow curves in Fig.~\ref{fig:detector_evaluation} show the performance for $L=8$, motivated by the fact that a \ac{CIR} at most spans a time interval of roughly $4\tau_\mathrm{RMS}=8\cdot 0.5\tau_\mathrm{RMS}$, see Fig.~\ref{fig:spectra}b).
Each data point is obtained by averaging the \ac{SER} over $10^6$ randomly transmitted symbols.

Across all \acp{VN}, for $L=2$, two groups of \ac{SER} curves are observed. Low-\ac{ISI} curves (round markers) approach zero as $N$ increases, since errors are mainly due to noise molecules, whose impact diminishes at high $N$.
High-\ac{ISI} curves (triangle markers) initially decrease with $N$ but approach an \ac{ISI}-induced error floor for large $N$, due to the detector's limited channel memory, preventing reliable \ac{ISI} estimation even for very large~$N$.

Notably, looking at the low-\ac{ISI} curves (round markers) for $L=2$, depending on the \ac{VN}, there are smaller or larger performance gaps between the three proposed sampling strategies.
For \ac{VN}1, the $t^\mathrm{peak}$ and $t^\mathrm{peak}_{g^*}$ strategies perform equally well (see \circnum{1}), because, in this \ac{VN}, the strongest path also dictates the global maximum of the \ac{CIR}, see Fig.~\ref{fig:spectra}b).
Sampling at $t_\mathrm{s}=\mathbb{E}[T]$ yields much worse performance, because $\mathbb{E}[T]$ lies approximately in the low-signal valley between the two path contributions of the \ac{VN}, see Fig.~\ref{fig:spectra}b).
In comparison, for \ac{VN}2, the gap between $t^\mathrm{peak}$ (or $t^\mathrm{peak}_{g^*}$) and $t_\mathrm{s}=\mathbb{E}[T]$ is reduced (see \circnum{3}), because the \ac{CIR} exhibits a more unimodal shape, see Fig.~\ref{fig:spectra}b).
For \ac{VN}3, $t^\mathrm{peak}$ and $t^\mathrm{peak}_{g^*}$ differ significantly, see Fig.~\ref{fig:spectra}b). 
As such, the three proposed sampling strategies exhibit largely different \acp{SER} in Fig.~\ref{fig:detector_evaluation}c), see \circnum{5}.
In summary, for the low-\ac{ISI} regime, we find that, while sampling at $t^\mathrm{peak}$ generally yields the best performance, \textit{depending on the \ac{VN} topology}, other sampling strategies ($t^\mathrm{peak}_{g^*}$ or $\mathbb{E}[T]$) may not be much worse in terms of \ac{SER}.

For \ac{VN}1 and \ac{VN}2, in the high-\ac{ISI} regime (triangle markers) for $L=2$, the performance ordering of the three sampling strategies is the same as in the low-\ac{ISI} case (see \circnum{2}, \circnum{4}): $t_\mathrm{s}=t^\mathrm{peak}$ performs best, followed by $t^\mathrm{peak}_{g^*}$ and $\mathbb{E}[T]$.
Moreover, all high-\ac{ISI} curves for \ac{VN}2 exhibit lower \acp{SER} than for \ac{VN}1 (see \circnum{4}, \circnum{2}).
The reason for this is that the \ac{CIR} of \ac{VN}2 lies more compactly in the $\mathbb{E}[T]\pm2\tau_\mathrm{RMS}$ interval compared to that of \ac{VN}1, see Fig.~\ref{fig:spectra}b).
Thus, the \ac{ISI} in \ac{VN}2 is slightly reduced compared to \ac{VN}1. 
The same is true for \ac{VN}3 (see \circnum{6}), which also generally exhibits lower \ac{ISI}-induced \acp{SER} than \ac{VN}1.
However, for \ac{VN}3, the performance ordering of the high-\ac{ISI} curves is different than for the other \acp{VN} (see \circnum{6}): $t^\mathrm{peak}_{g^*}$ performs best, followed by $\mathbb{E}[T]$ and $t^\mathrm{peak}$.
This ordering occurs for \ac{VN}3, because its \ac{CIR} carries a significant fraction of signal energy in the time interval following the global peak (more energy than, e.g., the tail of \ac{VN}1).
As a consequence, much of the signal energy of previous symbols impacts the detector in the form of \ac{ISI}, leading to higher \acp{SER}.
In contrast, when sampling at $t^\mathrm{peak}_{g^*}$ in \ac{VN}3, the tail beyond $t^\mathrm{peak}_{g^*}$ contains much less energy, leading to reduced \ac{ISI}.
Note that $t^\mathrm{peak}_{g^*}$ and $t^\mathrm{peak}$ yield similar \ac{CIR} amplitudes (see Fig.~\ref{fig:spectra}b)), i.e., this is not the deciding factor for the performance differences.

Lastly, across all \acp{VN}, the high-\ac{ISI} curves do not reach the performance of the low-\ac{ISI} curves even for large $L$ (compare triangle and circle marker yellow curves), since a larger $L=8$ improves \ac{ISI} estimation but simultaneously accumulates more noise across additional taps, i.e., increases the variance of the signal-dependent Poisson noise. This is confirmed by the thick green curve for $L=20$ in \ac{VN}3, which yields no performance improvement over $L=8$ (see \circnum{5}) in the high-\ac{ISI} regime.

In summary, the \ac{SER} results demonstrate that the \ac{MC} system design is topology-dependent, underscoring the usefulness of the  multipath channel metrics derived in Section~\ref{sec:Multi-Path_Metrics}:  
$\tau_\mathrm{RMS}$ directly informs the choice of $T_\mathrm{s}$ and $L$, while the \ac{VN} topology, e.g., reflected in $t_{g^*}^\mathrm{peak}$ and $\mathbb{E}[T]$, governs the choice of $t_\mathrm{s}$.

\section{Conclusion}\label{sec:Conclusion}
In this paper, we presented a focused treatment of \ac{MC} in \acp{VN} from a multipath communications perspective, inspired by classical tools from \ac{MWC}.
Leveraging the \ac{MIGHT} model, we derived closed-form expressions for a channel noise model, the classical multipath channel metrics \textit{\ac{RMS} delay spread}, \textit{mean excess delay}, and \textit{coherence bandwidth}, and the frequency response of the channel, which jointly characterize the severity of \ac{ISI} and the channel's frequency-selectivity. 
Exploiting the channel model and metrics, we proposed \ac{VN}-adapted sampling strategies and a \ac{DF} detector and illustrated the resulting communication performance in exemplary (large-scale) \acp{VN}.
We find that channel metrics borrowed from \ac{MWC} support the analysis of \ac{MC} in \acp{VN} and aid the choice of critical system parameters affecting \ac{ISI} and thus communication performance. 

Future work foresees the extension of the framework to multi-\ac{Tx}, multi-\ac{Rx} systems. Moreover, more \ac{VN}-adapted (sequence)-detectors could be considered and pulse-shaping at the \ac{Tx} could be investigated using the here proposed tools.

\appendix[Derivation of Noise Model]
Let $T_{p_i}$ denote the \ac{FPT} of a molecule released at time $t=0$ at the inlet of pipe $p_i$ (at $z_i=0$) and observed at its outlet located at\footnote{These spatial arguments are chosen without loss of generality.} $z=l_i$. 
Under advective-diffusive transport, the \ac{PDF} of $T_{p_i}$, denoted by $f_{T_{p_i}}(t_{p_i})$, is an \ac{IG}~\cite{Jakumeit2026}.
For a path $P_g$ comprising multiple pipes (and potentially bifurcations and junctions) in series, the path \ac{FPT} $T_g$ is equal to the sum of the individual pipe \acp{FPT} $T_{p_i}$~\cite[Eq.~(9)]{Jakumeit2026}, i.e.,
\begin{equation}
    T_g=\sum_{\{p_i\,\vert\, i\in\mathcal{E}_g\}} T_{p_i}\,.
\end{equation}
As the molecule transport in each pipe $p_i$ is modeled as an \textit{independent} advection-diffusion process~\cite{Jakumeit2026}, the \ac{PDF} of $T_g$ is obtained as~\cite[Eq.~(10)]{Jakumeit2026}
\begin{equation}\label{eqn:convolution}
    f_{T_g}(t_g)=\Asterisk_{\{p_i\,|\,i\in\mathcal{E}_k\}} f_{T_{p_i}}(t_{p_i}),
\end{equation}
where $\Asterisk_{\{p_i\,|\,i\in\mathcal{E}_k\}}$ denotes the convolution with respect to $t_{p_i}$ over all pipes contained in path $P_g$.
Since any given molecule released at the \ac{Tx} propagates through path $P_g$ and eventually arrives at the \ac{Rx} with probability $\gamma_{P_g}$, as dictated deterministically by the flow rates in the \ac{VN}, see~\eqref{eqn:path_fraction}, the probability of a molecule being located within the $\ac{Rx}$ in pipe $p_b$ at time $t_g=t$ is $\frac{l_\mathrm{Rx}}{\bar{u}_b}\gamma_{P_g}f_{T_g}(t)$ (under the \ac{UCA}).
Therefore, for $N$ molecules impulsively released by the \ac{Tx} at $t=0$ and propagating independently, the \ac{RV} $N_g(t)$, representing the number of molecules that traverse path $P_g$ and are observed within the \ac{Rx} at time $t$, is distributed as follows~\cite[Eq.~(70)]{Jamali2019}
\begin{equation}\label{eqn:binomial_distribution}
    N_g(t)\sim\mathrm{Binom}\left(N, \frac{l_\mathrm{Rx}}{\bar{u}_b}\gamma_{P_g}f_{T_g}(t)\right),
\end{equation}
where $\mathrm{Binom}(N,y)$ denotes
a binomially distributed \ac{RV} with $N$ trials and success probability $y$.
Thus, for each of the $N$ molecules, the event of being observed within the \ac{Rx} at time $t$ constitutes a Bernoulli trial with success probability 
$\frac{l_\mathrm{Rx}}{\bar{u}_b}\gamma_{P_g} f_{T_g}(t)$. 
Consequently, the number of observed molecules follows the binomial distribution in~\eqref{eqn:binomial_distribution}. 

Assuming the rare-event limit\footnote{The rare-event limit is justified for typical molecule counts and advection-diffusion parameters in single-duct \ac{MC} systems~\cite{Jamali2019}, and even more so in \acp{VN}, where signal splitting and increased temporal dispersion reduce the observation probability at the \ac{Rx}.}, i.e., small observation probability $\frac{l_\mathrm{Rx}}{\bar{u}_b}\gamma_gf_{T_g}(t_g)$ and a large number of released molecules $N$, the binomial distribution in~\eqref{eqn:binomial_distribution} is well-approximated by~\cite[Eq.~(74)]{Jamali2019}
\begin{equation}\label{eqn:path_poisson_distribution}
    N_g(t)\sim\mathrm{Pois}\left(N \frac{l_\mathrm{Rx}}{\bar{u}_b}\gamma_{P_g}f_{T_g}(t)\right)\,.
\end{equation}
Moreover, in~\cite{Jakumeit2026}, we show that the convolution in~\eqref{eqn:convolution} can be approximated by the \ac{IG} path flux in~\eqref{eqn:path_flux}, i.e., $f_{T_g}(t)\approx \bar{j}_g(t)$.

Lastly, since the individual numbers of molecules at the outlets of different paths $N_g(t)$ in the \ac{VN} are modeled as independent \acp{RV}, each of which is Poisson-distributed, and since the total number of molecules $N_\mathrm{Rx}(t)=\sum_{P_g\in\mathcal{P}_{\mathrm{Tx},\mathrm{Rx}}}N_g(t)$ observed within the \ac{Rx} at time $t$ is the sum of these independent \acp{RV}, $N_\mathrm{Rx}(t)$ is also Poisson-distributed due to the superposition property of the Poisson distribution. 
If we additionally assume the background noise molecules to be distributed according to a uniform concentration throughout the \ac{VN} such that, on average, $\bar{n}$ background noise molecules are observed within the \ac{Rx} at any time $t$, we obtain
\begin{equation}
    N_\mathrm{Rx}(t)\sim \mathrm{Pois} \Bigg(\underbrace{N\frac{l_\mathrm{Rx}}{\bar{u}_b}\sum_{P_g\in\mathcal{P}_{\mathrm{Tx},\mathrm{Rx}}}\gamma_g \bar{j}_g(t)+\bar{n}}_{\overset{\eqref{eqn:CIR_between_Tx_and_RX}}{=}Nh(t)+\bar{n}}\Bigg)\,.
\end{equation}
This result constitutes the continuous-time formulation of the noise model presented in Section~\ref{ssec:discrete_time_noise_model} for a single transmitted symbol 1 and thus completes the derivation.

\bibliography{bibliography}

\end{document}